\documentclass[prd,aps,eqsecnum,nofootinbib,amssymb,superscriptaddress,12pt]{revtex4}
\usepackage{amsmath}
\usepackage{amsthm}
\usepackage{mathtools}
\usepackage{graphicx}
\usepackage{amsfonts}

\usepackage{color, soul}

\usepackage{etoolbox}% http://ctan.org/pkg/etoolbox
\usepackage{lipsum}% http://ctan.org/pkg/lipsum
\usepackage[colorlinks = true]{hyperref}

\hypersetup{
	pdftitle = {Title},
	pdfauthor = {Authors}
}

\usepackage{xcolor}
\definecolor{darkred}  {rgb}{0.5,0,0}
\definecolor{darkblue} {rgb}{0,0,0.5}
\definecolor{darkgreen}{rgb}{0,0.5,0}
\hypersetup{
	urlcolor   = blue,         % color of external links
	linkcolor  = darkblue,     % color of internal links
	citecolor  = darkgreen,    % color of links to bibliography
	filecolor  = darkred       % color of file links
}

 \newtheorem{Def}{Definition}
 \newtheorem{Thm}{Theorem}
 \newtheorem{Lem}{Lemma}

\begin{document}
\title{Quantum phase transition from bounded to extensive entanglement entropy in a  frustration-free spin chain}
\author{Zhao Zhang, Amr Ahmadain and Israel Klich\\ Department of Physics, University of Virginia, Charlottesville, 22904, VA}
	
\begin{abstract}
We introduce a continuous family of frustration-free Hamiltonians with exactly solvable ground states. We prove that the  {ground state of our model is non-degenerate and exhibits} a novel  quantum phase transition from bounded entanglement entropy to a massively entangled state with volume entropy scaling. The ground state may be interpreted as a deformation away from the uniform superposition of colored Motzkin paths, showed by Movassagh and Shor \cite{movassagh2015power} to have a large (square-root) but sub-extensive scaling of entanglement into a state with an extensive entropy.
\end{abstract}

\maketitle	

\section{Introduction}

A fundamental question in quantum many-body physics, be it {in} condensed matter or in a  relativistic quantum field theory context, is how much entanglement can be generated by a reasonable Hamiltonian constructed from simple local terms. It is by now well known that generic states in the many-body Hilbert space are extensively entangled with respect to any partitioning of space \cite{page1993average,foong1994proof,sen1996average}. However, most familiar states arising as ground states of known Hamiltonians are significantly less entangled, and feature non-extensive entropy. Low entanglement entropy scaling is essential in our ability to study many quantum systems of interest using conventional computational means, which, in most cases, require splitting the state of the system into smaller blocks. Thus, the scaling of entanglement entropy in known ground states has been an area of intense research for the last two decades (For a recent review of entanglement in condensed matter systems see, e.g. \cite{laflorencie2015quantum}.).

Upper bounds on entanglement entropy scaling often take the form of an \textit{area law} which simply states that the entanglement entropy of subsystem $ A $ grows with the boundary of the entangling region rather than its volume. In one spatial dimension, the area law then implies that entanglement entropy is upper bounded by a constant independent of the size of $ A $. Entanglement scaling has an interesting, although not completely understood, relation with the scaling of the spectral gap,  and thus may give us some clues about the behavior of dispersion in various systems. In particular, for one-dimensional gapped systems, an area law was first proved by Hastings in \cite{hastings2007area}, and an improved bound was presented by Arad et al. in \cite{arad2013area}. 

While the area law has been shown to apply to a wide variety of gapped systems~(see e.g. \cite{eisert2010colloquium}), violations of the area law in the ground state of gapless phases have been also illustrated in several systems. For example, {$(1+1)$}-dimensional conformal field theories~\cite{callan1994high,holzhey1994high,calabrese2009entanglement} exhibit a logarithmic violation of the area law, while Fermi liquids \cite{wolf2006violation, gioev2006entanglement} exhibit logarithmic scaling of entanglement entropy in any dimension. In {one dimensional systems, even more severe violations of the area law have been recently exhibited~\cite{movassagh2015power,gottesman2010entanglement,irani2010ground,vitagliano2010volume,ramirez2014conformal,salberger2016fredkin}. In particular, Movassagh and Shor \cite{movassagh2015power} used a model based on colored Motzkin paths to describe a frustration-free spin-$d$ chain with a unique ground state where $d = 2s+1 $ is the local Hilbert space dimension and $s$ is the number of colors.  They demonstrate that for $s>1$, the entanglement entropy of half a chain, scales as a square root, $ O(\sqrt{n}) $, where  $2n$ is the number of spins. The uncolored case, $s=1$, was a introduced earlier in \cite{bravyi2012criticality}, where it was showed to have a logarithmic scaling $ O(\log{n}) $ of entanglement entropy. 

Motzkin walks, to be defined precisely below, are discrete walks on a lattice going from the origin $(0,0)$ to $(0,2n)$ without passing below the $x$ axis. These correspond to spin configurations where $\sum_{i=1}^m S^z_i\ge 0$ for $m<2n$ while $\sum_{i=1}^{2n} S^z_i= 0$. Loosely speaking, adding color means that a similar condition is satisfied for each `color', leading to further correlations between the two sides of the spin-chain. 

Adding the color degree of freedom is the main ingredient that allows for the enhanced entropy scaling compared to the uncolored model: the uniform superposition of Motzkin walks may be considered as a type of Brownian walk on a half space, whose typical displacement after $n$ steps (i.e. in the middle of the chain) is, therefore, $\sqrt{n}$. There are roughly $s^{\sqrt{n}}$ independent colorings for such a path that must be matched between the two halves of the walk, giving us the entropy $S_n\propto \sqrt{n}$~\cite{movassagh2015power}.

In this paper, we show how one can deform away from these models with a single parameter while maintaining the Hamiltonian frustration-free, non-degenerate, and translation invariant in the bulk. 
\begin{figure} \centering \includegraphics[width=0.7\textwidth]{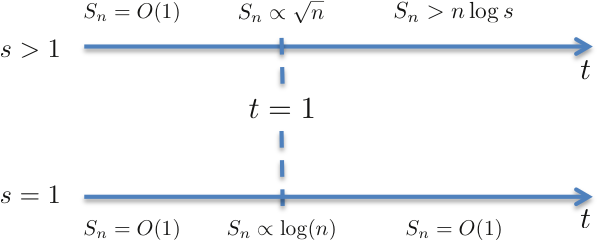} \caption{Entanglement entropy of the first $n$ sites in a chain of $2n$ sites for various phases of the colored and uncolored area-weighted Motzkin state.} \label{phases} \end{figure}
The resultant entropy behavior is depicted in Fig.~\ref{phases}:  For the colored model, we find a phase transition between the maximal scaling violation of the area law (i.e. volume scaling) into no violation at all (i.e. bounded entropy), with a transition through the special point discussed in \cite{movassagh2015power}. In addition, we find a transition between two regions with bound entropy through the critical point studied in \cite{bravyi2012criticality} in the uncolored chain. 

Our construction to increase/decrease entanglement entropy starts with the observation that Motzkin paths that reach a substantial height in the middle of the chain can contribute large color correlations between the chain halves. The idea is illustrated in Fig. \ref{HighLow}. 
\begin{figure} \centering \includegraphics[width=0.8\textwidth]{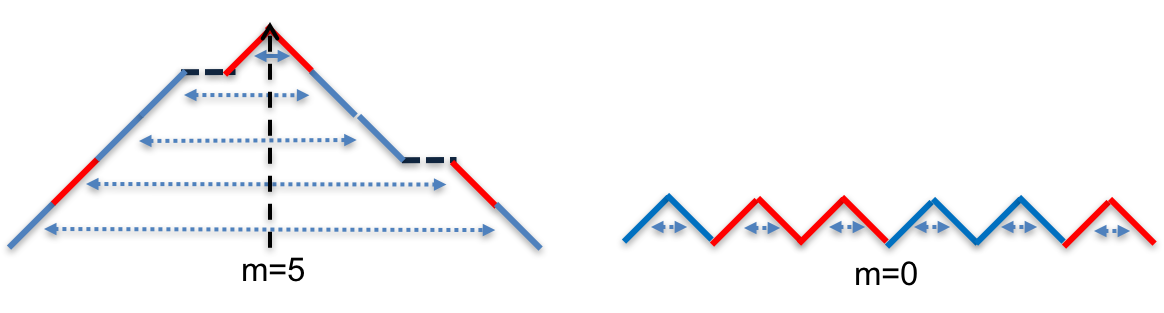} \caption{Motzkin paths require color correlations between the color of an up step and the color of the first step going down at the same height. The Motzkin path reaching height $m=5$ in the middle (left panel) contains more color correlations between the halves of the chain than a path of height $m=0$ (right panel). Favoring higher paths leads to more entanglement.} \label{HighLow} \end{figure}

Remarkably, a suitable wave function, containing a superposition of colored Motzkin paths which prefers high paths can be obtained as a frustration free and non-degenerate ground state of a Hamiltonian which is translational invariant in the bulk. Unfortunately, none of the ingredients in this statement are immediate. A generic change of the Hamiltonian presented in \cite{movassagh2015power} may very easily either break frustration freeness or the non-degeneracy condition nor will it increase entanglement.
Thus, we also need to show that we can do so in a way that the weight of these high paths is large enough as to overcome the contribution from more typical paths that reach height $\sqrt{n}$ at the middle of the chain \footnote{Typical paths of a colored Motzkin walk can be substantially more numerous even when counting possible colorings, at least for $t^2<s$}. 

In particular, here, the uniform superposition of the so-called Motzkin walks in the models of~\cite{bravyi2012criticality,movassagh2015power} is replaced in our model by a weighted superposition according to $t^{\text{area under the path}}$. Thus, higher paths are either exponentially preferred when $t>1$, or suppressed when $t<1$. A caricature of the resulting ground state is shown in Fig.~\ref{MotzkinSuperPos}. 
\begin{figure} \centering \includegraphics[width=0.4\textwidth]{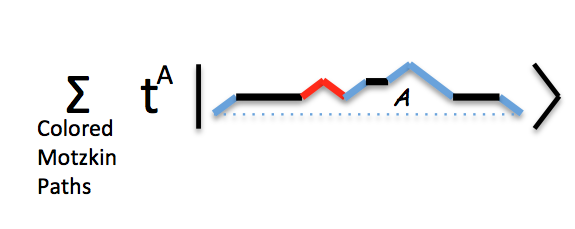} \caption{A caricature of the ground state of our model.} \label{MotzkinSuperPos} \end{figure}

Most of the paper is devoted to a rigorous demonstration that the above ideas can indeed yield the phase diagram shown in Fig.~\ref{phases}. 
In Section \ref{sec:background}, we give some background and describe related work on highly entangled ground states. We proceed in Section \ref{sec:mainres} to describe our construction for the case of a spin-1 chain, prove that the ground state is frustration free and non-degenerate, and extend the result to the model with an arbitrary number of colors. Section \ref{sec:scaling} is devoted to studying the behavior of the entropy associated with our model, and establish the phase diagram in Fig.~\ref{phases}. {We end with some open-ended questions and prospects for future work.} \newline

\subsection{Background and Related Work}\label{sec:background}
Several significant works on spin chains have achieved volume scaling of entanglement entropy~\cite{gottesman2010entanglement,irani2010ground,vitagliano2010volume,ramirez2014conformal}. These advances necessitated trading off translation-invariance, the non-degeneracy of the ground state or using a particularly large number of states per lattice site.  Our new model which is unique in achieving a controlled and intuitively transparent phase transition between volume scaling of entanglement entropy and bounded entropy while simultaneously preserving bulk translation-invariance and uniqueness of the ground state. The first model exhibiting volume scaling of entanglement entropy in a 1D spin chain model was introduced by Irani in \cite{irani2010ground}. This model is frustration-free and translation-invariant and achieves linear scaling of entanglement entropy on some regions of the spin chain; it does so at the expense of having a particularly large local Hilbert space dimension of $d=21$. Independently, Gottesman and Hastings \cite{gottesman2010entanglement}  presented a model of 1D spin chain with an outcome similar to that of Irani in \cite{irani2010ground} where they succeeded in showing linear scaling of entanglement entropy for some blocks of the chain but with a smaller number of states per lattice site, i.e. $d=9$. They did so, however, by explicitly breaking the translation-invariance of the system. Vitagliano et al in \cite{vitagliano2010volume} proposed a model of a frustration-free spin-$1/2$ chain with nearest-neighbor interactions where the entanglement entropy of the ground state scales according to a volume law entropy. The authors in \cite{vitagliano2010volume} achieved volume scaling by explicitly breaking translational invariance and by using real-space RG approach to find a carefully fine-tuned set of coupling constants for the inhomogeneous XX (free fermion) model Hamiltonian. Ramirez et al in \cite{ramirez2014conformal}  generalized the model in \cite{vitagliano2010volume} for 1D spin-$1/2$ critical Hamiltonians by finding, using real-pace RG techniques, a set of exponentially decaying coupling constants that allow the violation of the area law by volume scaling of the entanglement entropy.  

Another recent example is the work by Salberger and Korepin \cite{salberger2016fredkin}  in which they constructed a model of interacting spin-$1/2$ chain that generalizes the work in ~\cite{movassagh2015power} by using Dyck walks instead of Motzkin walks and by expressing the Hamiltonian in terms of Fredkin gates. For their model, the authors in~\cite{salberger2016fredkin} were able to show $ O(\sqrt{n})$ scaling of the entanglement entropy.

In this context, it is important to point out that, as shown in~\cite{movassagh2010unfrustrated}, there are three distinct regimes for Hamiltonians of 1D spin chains whose terms are generic local projectors of fixed rank $ r $: (i) When $ r > d^2/4 $, the Hamiltonian is frustrated for sufficiently large spin chains and analytical as well as numerical work showed that no zero-energy ground states exist, (ii) a regime where $d \leq r \leq d^2/4 $ where many zero-energy ground states are allowed analytically and where numerical investigation suggests that they all carry a large amount of entanglement, and a (iii) frustration-free regime with $ r<d $ where the ground states can be represented by a matrix product state. The Motzkin path-based models first introduced in \cite{bravyi2012criticality} and later with the addition of color in \cite{movassagh2015power} represent a special case where the Hamiltonian turned out to be frustration-free for the special case of $ r=d=3 $. For this reason, the authors in~\cite{bravyi2012criticality} have pointed out that any arbitrary small deformations of the projectors in the Motzkin path Hamiltonian will make it generic and thus throw its ground state into the frustrated regime. In the model presented in this paper, however, we derive a simple equation that relates the weights of local moves at different sites of the chain and thus deform the local projectors away from the uniform case in such a way that frustration-freeness of the Hamiltonian is maintained. 

\section{The Main Result} \label{sec:mainres}
The following theorems are the central result of our work.
\begin{Thm} The following Hamiltonian, acting on a $2n$ sites of a spin-s chain,
 \begin{equation} H(s,t) = \Pi_{boundary}(s) + \sum_{j=1}^{2n-1}\Pi_{j,j+1}(s,t) +\sum_{j=1}^{2n-1} \Pi_{j,j+1}^{cross}(s), \label{cHam} \end{equation} where  \begin{align*} \Pi_{boundary}(s) &= \sum_{k=1}^{s}(|r^k\rangle\langle r^k|_1 + |l^k\rangle\langle l^k|_{2n}) ,\\ \Pi_{j,j+1}(s,t) &= \sum_{k=1}^{s}(|\Phi(t)^k\rangle\langle\Phi(t)^k|_{j,j+1} + |\Psi(t)^k\rangle\langle\Psi(t)^k|_{j,j+1} + |\Theta(t)^k\rangle\langle\Theta(t)^k|_{j,j+1}),\\ \Pi_{j,j+1}^{cross}(s) &= \sum_{k \neq k'}|l^kr^{k'}\rangle\langle l^kr^{k'}|,\end{align*} with \begin{align*} |\Phi^k(t)\rangle &= \frac{1}{\sqrt{1+t^2}}(|l^k0\rangle - t|0l^k\rangle), \\ |\Psi^k(t)\rangle &= \frac{1}{\sqrt{1+t^2}}(|0r^k\rangle - t|r^k0\rangle), \\ |\Theta^k(t)\rangle &= \frac{1}{\sqrt{1+t^2}}(|l^kr^k\rangle - t|00\rangle), \end{align*} has a unique zero energy ground state 
\begin{equation} |GS\rangle = \frac{1}{\mathcal{N}}\sum_{\substack{w \in \{s-colored\\ Motzkin\ walks\}}} t^{\mathcal{A}(w)} |w\rangle, \label{cgs} \end{equation} where $\mathcal{A}(w)$ denotes the area below the Motzkin walk $w$, and $\mathcal{N}$ is a normalization factor. \label{Main model}\end{Thm}

The phase diagram in Fig.~\ref{phases} is a consequence of the following theorem, with entanglement entropy of the half chain $S_n$:
\begin{Thm} 
The wave function $|GS\rangle$ above has the following behavior of entanglement entropy of half a chain:
\begin{align*}
S_n = \begin{cases} O(n) & \text{if $t>1,s>1$} \\
O(1) & \text{if $t<1$} 
\end{cases}
\end{align*}
\end{Thm}
\noindent The results for the $t=1$ point are described in \cite{bravyi2012criticality,movassagh2015power}.

\subsection{A Frustration Free Deformation}\label{sec:setup}
For simplicity, we start our derivation from the uncolored model. 
In this section we explain how to get a deformation of the spin-$1$ Hamiltonian described in \cite{bravyi2012criticality} while maintaining frustration freeness.

We start by quickly reviewing the construction in \cite{bravyi2012criticality}. The ground state of this a spin chain of length $2n$, can be represented as an equal weight superposition of `Motzkin walks', defined as follows:
\begin{Def} \label{def: Motzkin} A Motzkin walk (or path) on 2n steps is any path from (0, 0) to (0, 2n) with steps (1, 0), (1, 1) and (1, --1) that never pass below the x-axis. \end{Def}

Pictorially, a Motzkin walk corresponds to a mountain range that is located between site $0$ and $2n$. A Motzkin walk naturally encodes a spin-$1$ state $|\sigma_{1},...\sigma_{2n}\rangle$ constructed by assigning for the local spin variables $\sigma_{k}=+1,-1$ or $0$ if the walk goes up, down, or stays flat at site $k$. 
As depicted in Fig.~\ref{MotzkinUncolored}, Motzkin paths can be thought of as grammatically allowed choices for placing left and right parentheses in a sentence, since a right parenthesis is only allowed to be placed if there is an unpaired left parenthesis to the its left. Thus, following the notation of \cite{bravyi2012criticality} we will span the local spin basis using $|l\rangle$,  $|r\rangle$ and $|0\rangle$, corresponding to the $S^{z}=+1,-1$, and $0$  states respectively. 

\begin{figure} \centering \includegraphics[width=0.4\textwidth]{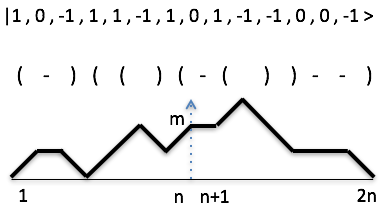} \caption{A spin-$1$ configuration, and its representation as a set of parentheses and a Motzkin path.} \label{MotzkinUncolored} \end{figure}

The superposition of Motzkin walks given by
\begin{eqnarray} |GS\rangle = \frac{1}{\mathcal{N}}\sum_{\substack{w \in \{Motzkin\ walks\}}} |w\rangle,\end{eqnarray} where $\mathcal{N}$ is normalization factor, is the unique ground state of the frustration free Hamiltonian 
\begin{eqnarray} \label{BravyiHamiltonian} 
H = |r\rangle\langle r|_1 + |l\rangle\langle l|_{2n} + \sum_{j=1}^{2n-1}\pi_{j,j+1}, \end{eqnarray}  
where projectors acting on spin $j, j+1$ \begin{equation*} \pi_{j,j+1} = |\phi\rangle\langle\phi|_{j,j+1} + |\psi\rangle\langle\psi|_{j,j+1} + |\theta\rangle\langle\theta|_{j,j+1}, \end{equation*} with \begin{equation*} |\phi\rangle = \frac{1}{\sqrt{2}}(|0l\rangle - |l0\rangle), |\psi\rangle = \frac{1}{\sqrt{2}}(|0r\rangle - |r0\rangle), |\theta\rangle = \frac{1}{\sqrt{2}}(|00\rangle - |lr\rangle). \end{equation*}

Our aim in this section is to deform away from the Hamiltonian (\ref{BravyiHamiltonian}) into a Hamiltonian with a unique ground state that is a weighted superposition of Motzkin paths, that can favor paths with greater height in the middle of the mountains, while preserving the frustration free nature. This is accomplished in the following theorem: \\

\begin{Thm} \label{thm: FF} The Hamiltonian 
\begin{equation} H = |r\rangle\langle r|_1 + |l\rangle\langle l|_{2n} + \sum_{j=1}^{2n-1}\Pi_{j,j+1}, \label{Ham}\end{equation} 
defined on $(\mathbb{C}^3)^{\otimes 2n}$ 
with  \begin{equation*} \Pi_{j,j+1} = |\Phi\rangle\langle\Phi|_{j,j+1} + |\Psi\rangle\langle\Psi|_{j,j+1} + |\Theta\rangle\langle\Theta|_{j,j+1}, \end{equation*} where \begin{align*} |\Phi\rangle_{j,j+1} &= \cos\phi_{j+\frac{1}{2}} |0l\rangle_{j,j+1} - \sin\phi_{j+\frac{1}{2}}|l0\rangle_{j,j+1},\\ |\Psi\rangle_{j,j+1} &= \cos\psi_{j+\frac{1}{2}}|0r\rangle_{j,j+1} - \sin\psi_{j+\frac{1}{2}}|r0\rangle_{j,j+1},\\ |\Theta\rangle_{j,j+1} &= \cos\theta_{j+\frac{1}{2}}|00\rangle_{j,j+1} - \sin\theta_{j+\frac{1}{2}}|lr\rangle_{j,j+1} \end{align*} is frustration free and has a unique ground state with zero energy provided $\psi_{i}, \phi_{i},\theta_{i}\in (0,\pi/2) $ satisfy relations 
\begin{equation} \tan\theta_i\cot\phi_i = \tan\theta_{i+1}\tan\psi_{i+1}, \qquad i=\frac{3}{2}, \frac{5}{2}, \frac{7}{2}, \ldots, 2n-\frac{1}{2}. \label{tune}\end{equation} \end{Thm}

We first prove the uniqueness of the ground state (GS) first assuming the Hamiltonian is frustration free, and then show that the Hamiltonian is indeed frustration free under condition \eqref{tune}.

{\it Remark:} When some of the angles equal an integer multiple of $\pi/2$, the Hamiltonian is still frustration free, but may have a degenerate ground state.

\begin{proof} (Uniqueness of GS) We look for a frustration-free ground state that will be annihilated by each of the terms in the Hamiltonian (\ref{Ham}).  We define the following $R, L, F$ moves and their inverses: \begin{equation} |l0\rangle \xrightleftharpoons[R^{-1}]{R} \tan\phi|0l\rangle,\quad |0r\rangle \xrightleftharpoons[L^{-1}]{L} \cot\psi|r0\rangle, \quad |lr\rangle \xrightleftharpoons[F^{-1}]{F} \tan\theta|00\rangle. \label{moves} \end{equation}
We first note that if a ground state wave function contains a particular spin configuration (a ``walk''), then the ground state wave function must contain as well a superposition of all states which can be obtained from it by the set of moves (\ref{moves}).

Indeed, at each neighboring two sites, the local spin state can be one of the nine possible configurations in $\{ |ll\rangle, |rr\rangle, |rl\rangle, |0l\rangle, |l0\rangle, |0r\rangle, |r0\rangle, |00\rangle, |lr\rangle\}$, the first 3 of which are annihilated by the projectors $\Pi$ individually. The rest must form pairs $\sin\phi |0l\rangle + \cos\phi|l0\rangle, \sin\psi|0r\rangle + \cos\psi|r0\rangle$, and $\sin\theta|00\rangle + \cos\theta|lr\rangle$ to be annihilated by $|\Phi\rangle\langle\Phi|, |\Psi\rangle\langle\Psi|$, and $|\Theta\rangle\langle\Theta|$ respectively. Each of these  superpositions corresponds to mixing between states related by the moves $R, L$, and $F$.  

The processes of generating {additional walks starting from a given one is `mixing'} in that it can keep going on and on until all Motzkin walks are included in the superposition. To see this we construct the following procedures of relating Motzkin walks to the `flat' mountain, i.e. the string of spins $000\ldots 0$:
If the highest peak of the current mountain is of the type $l0\cdots r$ (i.e. a plateau), then keep applying $L$ and/or $R$ moves until it becomes $lr$ (i.e. a hill), otherwise apply the $F$ operation to the hill. Note that there are multiple choices of numbers and orders of $L$ and $R$ moves applied to make a plateau a hill. This can be done iteratively until the mountain is completely flat. (See Fig.~\ref{uni}.) Given any Motzkin walk, we can represent it by a sequence of consecutive moves applied to it to get to the flat mountain, e.g. $|M_1\rangle = M_1|000\ldots 0\rangle = (F_{i_n}\cdots R_{i_3}L_{i_2}L_{i_1})^{-1}|000\ldots 0\rangle$. So any two Motzkin walks are related by $|M_2\rangle = M_2 M_1^{-1}|M_1\rangle$. 

Similarly, it is easy to see that any walk which crosses below zero, or that does not return to zero at the end of the chain, can be transformed by the $R,L,F$ moves and their inverses into a walk that violates the boundary projectors.
Therefore if a zero energy ground state does exist, then it will be the unique superposition of all Motzkin walks with weights determined by the tuned projectors. \end{proof}

\begin{figure} \centering \includegraphics[width=0.5\textwidth]{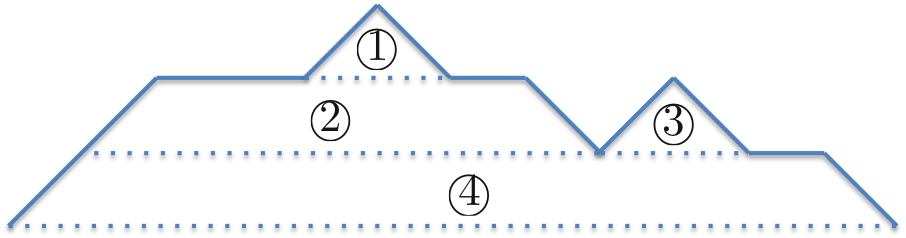} \caption{Iterative procedure to flatten a mountain to the ground, where steps 2 and 3 can be interchanged.} \label{uni} \end{figure}

It remains to be examined whether the aforementioned freedom in choosing the sequences of moves may result in ambiguities in the relative weights between Motzkin walks. It turns out that the tuning conditions \eqref{tune} suffice to guarantee that a superposition of Motzkin walks can be written without ambiguities in the relative amplitudes. It can be seen from an observation of the local moves involving three adjacent sites illustrated in Fig.~\ref{commu}. The two ways to get $|000\rangle$ from $|l0r\rangle$ will give the same relative weight if and only if the mixing angles at two neighboring junctions satisfy the relation \eqref{tune}. The global version of this statement holds as well:

\begin{figure} \centering \includegraphics[width=0.5\textwidth]{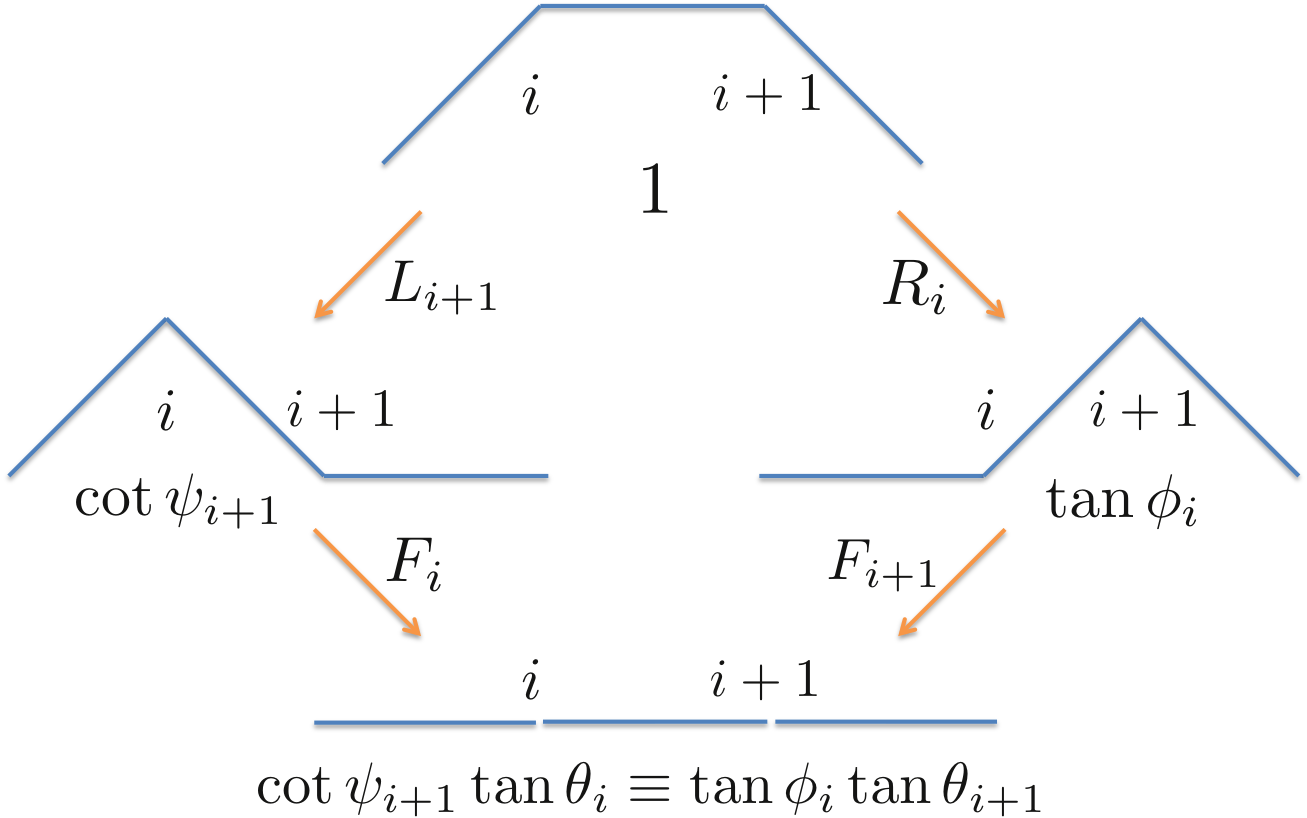} \caption{Two different sequences of moves to relate local state $|000\rangle$ to $|l0r\rangle$, and the relative weights of each state invovled.} \label{commu} \end{figure}

\begin{figure} \centering \includegraphics[width=0.5\textwidth]{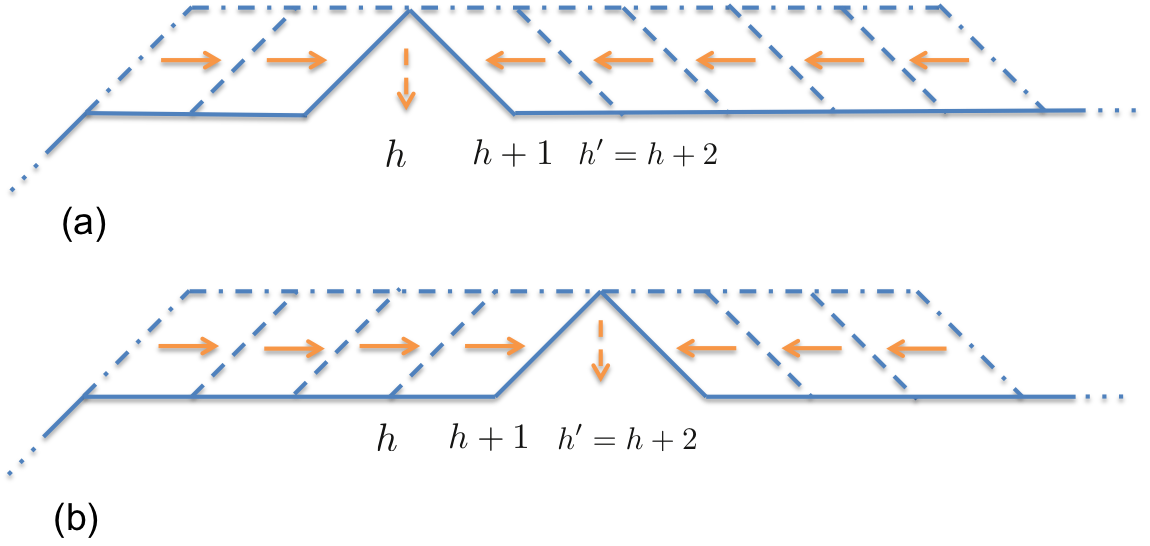} \caption{Different sequences of $L, R, F$ moves to get a hill (solid) from a starting plateau (dash dot) with intermediate plateaus after each move (dash). (a) Two different sequences of moves $R_{h-1}L_{h+1}L_{h+2}L_{h+3}R_{h-2}L_{h+4}L_{h+5}$ and $L_{h+1}L_{h+2}R_{h-1}L_{h+3}L_{h+4}L_{h+5}R_{h-2}$ with the same location of hill always give the same relative weight. (b) A sequence of moves with hill location different from those in (a) could generically give a different relative weight, except when relation~\eqref{tune} is satisfied.} \label{loc} \end{figure}

\begin{proof} (Frustration Freeness) A plateau of width $d$, (that is, the number of $0$ spins), is generated by one hill (or $F^{-1}$ move) and $d$ $R^{-1}$ and $L^{-1}$ moves. (See Fig.~\ref{loc}.) Once the location of the hill is chosen, $R^{-1}$ ($L^{-1}$) only acts on its left (resp. right), and whether acting an $R^{-1}$ on the left first or an $L^{-1}$ on the right first doesn't affect the weight. So the weights are completely determined by the location of the hills that plateaus originate from at each level. The weights of the same plateau generated by hills at location $h$ and $h'$ are related by \begin{equation} m(h) = \prod_{i=h}^{h'-1}\frac{\tan\theta_i \cot\phi_i}{\tan\theta_{i+1} \tan\psi_{i+1}} m(h') = m(h').\end{equation} Therefore the weight of each mountain is an invariant of the sequence of moves chosen to construct it from the flat mountain.\\  Furthermore, if two mountains are related directly to each other without passing through the flat mountain, by a sequence of $N$ moves, then each move in the sequence can be viewed as either a `piling' move away from the flat mountain or a `flattening' move towards it. So the intermediate mountains generated in this sequence each have definite weight $m_1, m_2, m_3, \ldots$, and the relative weight between these two mountains \begin{equation} \frac{m_A}{m_B} = \frac{m_A}{m_1}\frac{m_1}{m_2}\frac{m_2}{m_3}\cdots \frac{m_N}{m_B}\end{equation} is an invariant. It follows that the relative weights between any two Motzkin walks are well-defined and conditions~\eqref{tune} is sufficient for Hamiltonian~\eqref{Ham} to be frustration free.\end{proof}

\subsection{The Colorful Model}\label{sec:model}
In this section we establish our main model as summarized in theorem \ref{Main model}.

Following \cite{movassagh2015power} we add color to the Motzkin paths in the ground state superposition. In this case one can think of the states as the admissible ways of placing parentheses (labeled by color) of several types into a sentence. In the $s$-colored model, the local spin space is $2s+1$ dimensional and spanned by the basis states $|0\rangle,|l^{1}\rangle,..|l^{s}\rangle,|r^{1}\rangle,..|r^{s}\rangle$. For illustration, Fig.~\ref{MotzkinPathColored} depicts a coloring choice for the uncolored Motzkin path in Fig. \ref{MotzkinUncolored} and the associated spin state.
\begin{figure} \centering \includegraphics[width=0.8\textwidth]{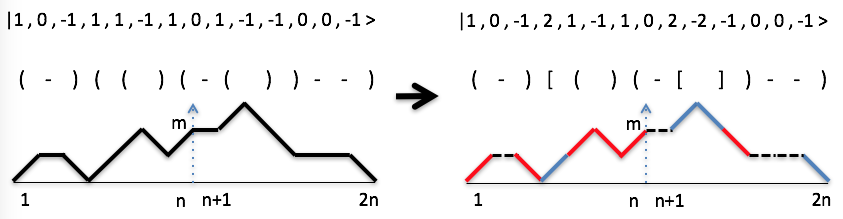} \caption{Coloring a Motzkin path leads to a higher spin configuration. Representations of a spin-2 state as a a set of parentheses and as Motzkin path, where $ 1,-1\leftrightarrow (,)  \leftrightarrow$ red and $ 2,-2\leftrightarrow [,]  \leftrightarrow$ blue.} \label{MotzkinPathColored} \end{figure}

Below we incorporate the colors in our model, and pick a particular, translationally invariant choice for the angles in Eq.~\eqref{tune}, $\cot\phi_i = \tan\psi_i = t$. For simplicity, we further let $\cot\theta_i = t$. Now, all three moves $R^{-1}, L^{-1}$, and $F^{-1}$ in (\ref{moves}) change the weight of a Motzkin path by a factor of $t$ and increase the area below the mountain by exactly one unit. So the weight of each mountain is simply determined by the area below it. Thus, the weight of each mountain compared to the flat Motzkin path is given by $ t^{\mathcal{A}(w)}$. The result is a ground state where the Motzkin paths are exponentially weighted according to the area under paths, rather than a uniform superposition.
The result is summarized in theorem \ref{Main model}.

\section{Entanglement entropy}\label{sec:scaling}
In the Schmidt decomposition of the ground state~\eqref{cgs}, the coloring of the unpaired spins in the second half of the system is completely determined by that in the first half. We can write the decomposition as:
\begin{equation} |GS\rangle = \sum_{m=0}^{n} \sqrt{p_{n,m}} \sum_{x\in \{l^1,l^2,\ldots,l^s\}^m} |\hat{C}_0,m,x\rangle_{1,\ldots,n} \otimes |\hat{C}_m,0,\bar{x}\rangle_{n+1,\ldots,2n}, \end{equation} where $|\hat{C}_p,q,x\rangle_{1,\ldots,n} $  is a weighted superposition states in $\{0, l^1, \ldots, l^s, r^1,\ldots, r^s\}^n$ with $p$ excess right, $q$ excess left parentheses and a particular coloring $x$ of the unmatched parentheses, such that $\langle GS|(|\hat{C}_0,m,x\rangle_{1,\ldots,n} \otimes |\hat{C}_m,0,\bar{x}\rangle_{n+1,\ldots,2n}) \neq 0$, and $\bar{x}$ is the coloring in the second half of the chain that matches $x$. The decomposition gives the Schmidt number \begin{equation} p_{n,m}(s,t) = \frac{M_{n,m}^2(s,t)}{N_n(s,t)}, \end{equation} where \begin{align}M_{n,m}(s,t) &\equiv \sum_{i=0}^{\frac{n-m}{2}} s^i\sum_{\substack{w\in \{1st\ half\ of\ Motzkin\ walks\ with \\  i\ paired\ spins \ stopped\ at\ (n,m)\}}}t^{\mathcal{A}(w)},\\ N_n(s,t) &\equiv \sum_{m=0}^{n}s^mM_{n,m}^2(s,t). \end{align} And the entanglement entropy of the half chain in the ground state is given by \begin{equation} S_n(s,t) = -\sum_{m=0}^ns^mp_{n,m}(s,t)\log p_{n,m}(s,t).\end{equation}

First we notice that the Hamiltonian \eqref{cHam} and the ground state \eqref{cgs} reproduce those by Movassagh and Shor when $t=1$. So the entanglement entropy $S_n(s,1)$ scales as $\sqrt{n}$. To study the asymptotic scaling of $S_n(s,t)$ with the system size when $t \neq 1$, we need the following lemma about the behavior of $M_{n,m}$ as a function of $m$.

\begin{Lem} $M_{n,m}$ satisfies the following recurrence relations \begin{align} M_{k+1,k+1} &= t^{k+\frac{1}{2}}M_{k,k},\\ M_{k+1,k} &= t^kM_{k,k} + t^{k-\frac{1}{2}}M_{k,k-1},\\ M_{k+1,m} &= st^{m+\frac{1}{2}}M_{k, m+1} + t^mM_{k,m} + t^{m-\frac{1}{2}}M_{k,m-1}, \quad 1<m<k-1,\\ M_{k+1,0} &= st^{\frac{1}{2}}M_{k,1} + M_{k,0}. \label{recurrel}\end{align} \end{Lem}

\begin{proof} This can be easily seen from the possible ways to arrive at a certain destination and the increment of area below each path as illustrated in Fig.~\ref{fig: recur}. \end{proof}

\begin{figure} \centering \includegraphics[width=0.5\textwidth]{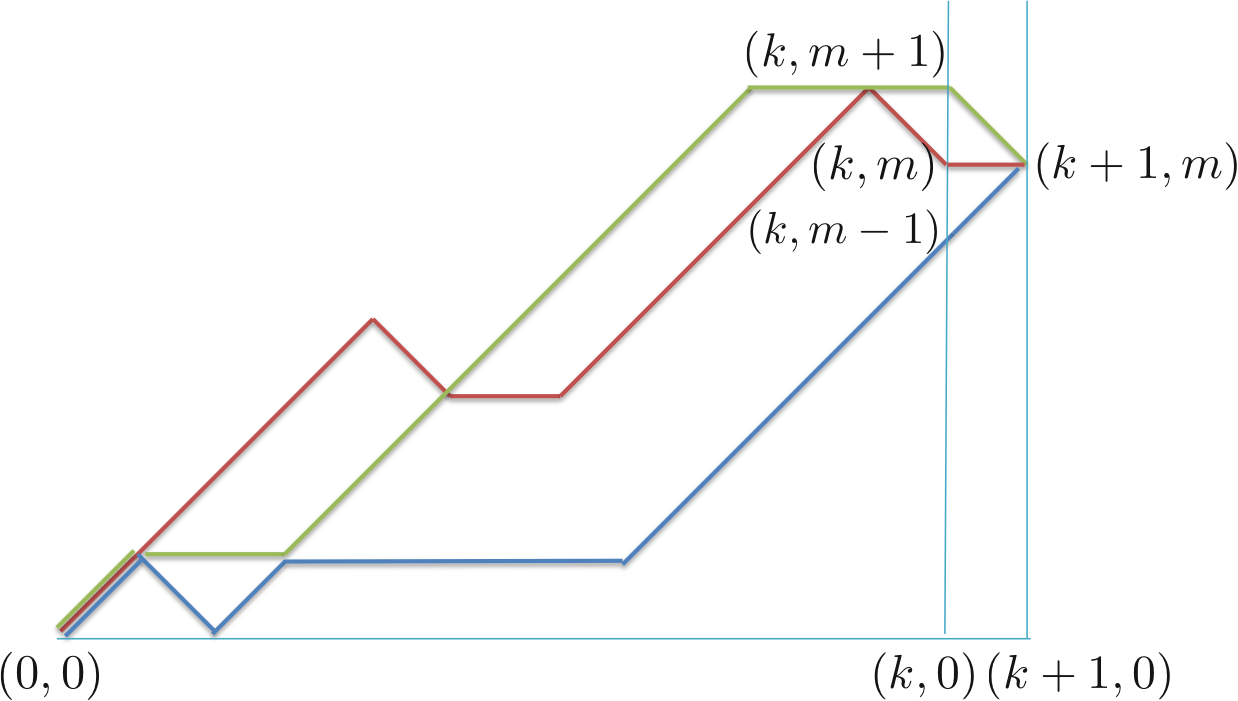} \caption{Three representative paths of different heights at position $k$ to generate a path that ends with height $m$ at position $k+1$. The increments in the area below the paths are the areas sandwiched between two vertical lines, which are $m+\frac{1}{2}, m, m-\frac{1}{2}$ respectively for red, green and blue paths.} \label{fig: recur} \end{figure}

Starting from the seed values $M_{0,0} = 1$, by using the recurrence relations repetitively, one can calculate the values of $M_{n,m}$ and Schmidt numbers for any $m$ and calculate the entanglement entropy. In the next sections we show how to get a lower bound on the entanglement entropy when $t>1$ and an upper bound when $t<1$.

\subsection{$t>1,s>1:$ Volume scaling of entropy}
In this section we prove the linear scaling of entropy for $t,s>1$ as summarized by:
\begin{Thm}  In the state \eqref{cgs}, when $t > 1$, the entanglement entropy of sites $1...n$, is bounded below as  $S_{n}>n \log s +const.$ for all $n$, where $const.$ is an $n$ independent constant. \label{Thm:linearEntropy}\end{Thm} 

The presence of large entropy is a consequence of contributions from the possible colorings of `high' Motzkin paths, those with height at the middle scaling as $O(n)$. Any coloring of the ascending part of the path on the left half chain will have high correlations with the coloring on the descending part of the path on the right, giving us $s^{n}$ distinct left-right color-correlated states in the superposition.

To get a better handle of the type of distribution the recursion relations \eqref{recurrel} lead to, we find it convenient to view the recursion evolution as a process of increasing/decreasing $m$ while $n$ is viewed as discrete `time'. We will show below that for large enough $n$, the distribution associated with the $M_{n,m}$ essentially propagates ballistically (as function of $n$), with very little spread. This property will establish that the typical hight at the middle of a $2n$ chain scales linearly with $n$.

Before exhibiting the proof, we need to develop a few preliminary steps. 
We encode the distributions $M_{n,m}$ as coefficients of wavefunctions defined on the set $|m\rangle, m = 0,1,2,\ldots$ as \begin{equation} |\mathcal{M}_n\rangle = \sum_{m=0}^\infty M_{n,m} |m\rangle, \qquad M_{n,m} = 0\ \text{if}\ m>n.\end{equation} 
We define the following `shift' and `height' operators, which we will use in describing  the `evolution' of the distribution $M_{n,m}$ as function of `time' $n$.
\begin{align} \mathcal{S} |m\rangle &= |m-1\rangle, \qquad |-1\rangle = 0; \\ \mathcal{H} |m\rangle &= m |m\rangle. \end{align} Explicitly, $t^{\mathcal{H}}, \mathcal{S}, \mathcal{S}^\dagger$ 
act on $|\mathcal{M}_n\rangle$ as follows. \begin{align} t^{\mathcal{H}} |\mathcal{M}_n\rangle &= \sum_{m=0}^\infty M_{n,m} |m\rangle, \\ \mathcal{S} |\mathcal{M}_n\rangle &= \sum_{m=1}^\infty M_{n,m} |m-1\rangle = \sum_{m=0}^\infty M_{n,m+1} |m\rangle,\\ \mathcal{S}^\dagger |\mathcal{M}_n\rangle &= \sum_{m=0}^\infty M_{n,m-1} |m\rangle. \end{align}
As remarked above, we aim to show that for large enough $n$, $|\mathcal{M}_{n+1}\rangle \propto \mathcal{S}^\dagger |\mathcal{M}_n\rangle$, describing essentially ballistic propagation with no spread. 
For the recurrence relation \eqref{recurrel} to be satisfied, we require 
\begin{align*}& M_{n+1,m} = \langle m | \mathcal{M}_{n+1}\rangle = s t^{m + \frac{1}{2}} \langle m+1|\mathcal{M}_n\rangle + t^m \langle m |\mathcal{M}_n\rangle + t^{m - \frac{1}{2}} \langle m-1|\mathcal{M}_n\rangle \\ &= \langle m | s t^{\mathcal{H} + \frac{1}{2}} \mathcal{S} + t^{\mathcal{H}} + t^{\mathcal{H} - \frac{1}{2}} \mathcal{S}^\dagger |\mathcal{M}_n\rangle. \end{align*} Therefore, \begin{equation} |\mathcal{M}_{n+1}\rangle = t^{\mathcal{H}} (s\sqrt{t}\mathcal{S} + 1 + \frac{1}{\sqrt{t}}\mathcal{S}^\dagger)|\mathcal{M}_n\rangle . \end{equation} Using the relations, 
\begin{align} t^{\mathcal{H}} \mathcal{S} &= t^{-1} \mathcal{S} t^{\mathcal{H}},\\ t^{\mathcal{H}} \mathcal{S}^\dagger &= t \mathcal{S}^\dagger t^{\mathcal{H}}, \end{align} we have \begin{equation} t^{k\mathcal{H}} (s\sqrt{t}\mathcal{S} + 1 + \frac{1}{\sqrt{t}}\mathcal{S}^\dagger) = (st^{-(k - \frac{1}{2})} \mathcal{S} + 1 + t^{k - \frac{1}{2}} \mathcal{S}^\dagger) t^{k\mathcal{H}}, \end{equation} and \begin{align} |\mathcal{M}_n\rangle &= [t^{\mathcal{H}} (s\sqrt{t}\mathcal{S} + 1 + \frac{1}{\sqrt{t}}\mathcal{S}^\dagger)]^n |\mathcal{M}_0\rangle \\ &= \vec{\mathcal{K}}\prod_{k=1}^n  (st^{-(k - \frac{1}{2})} \mathcal{S} + 1 + t^{k - \frac{1}{2}} \mathcal{S}^\dagger) |0\rangle, \label{korder}\end{align} where $\vec{\mathcal{K}}$ denotes ordering the multiplications in the product such that factors with greater $k$ value is on the right. It is now evident that the factors in the product above are dominated by the $\mathcal{S}^\dagger$ term for large $k$, giving us ``ballistic'' evolution with $n$. To be more precise:

\begin{Lem} Let $m^*$ be such that $\sup_{m} M_{n,m} = M_{n,m^*}$, then $\exists N_{0} < n$, such that when $t > 1$, $m^* \in [n - 2N_{0}, n]$. \label{Lemma:mstar}\end{Lem}

\begin{proof} Let 
\begin{equation} |\mathcal{M'}_n\rangle =\vec{\mathcal{K}}\prod_{k={N_{0}}+1}^{n} (st^{-(k - \frac{1}{2})} \mathcal{S} + 1 + t^{k - \frac{1}{2}} \mathcal{S}^\dagger) |0\rangle. \end{equation}
%then \begin{align*} M'_{n, n-N} &\equiv \langle n - N|\mathcal{M'}\rangle = t^{\sum_{k=N+1}^{n} k-\frac{1}{2}} = t^{\frac{n^2-N^2}{2}},\\ M'_{n, n - N -1} &\equiv \langle n - N - 1|\mathcal{M'}\rangle = t^{\frac{n^2-N^2}{2}} \sum_{k=N+1}^{n} t^{-(k-\frac{1}{2})}\\ &= t^{\frac{1}{2}}\frac{t^{-N} - t^{-n}}{t-1} M'_{n, n-N}\\ &< t^{-N} \frac{t^{\frac{1}{2}}}{t-1} M'_{n, n-N}. \end{align*} 
Note that
\begin{equation} t^{-(k-\frac{1}{2})}\| 1 + st^{-(k-\frac{1}{2})} \mathcal{S} \| \le t^{-(k-\frac{1}{2})} + st^{-2(k-\frac{1}{2})} \equiv c_k,\end{equation} 
so that:
\begin{eqnarray*} &
||st^{-(k - \frac{1}{2})}\mathcal{S} + 1 + t^{k - \frac{1}{2}} \mathcal{S}^\dagger)||<t^{k - \frac{1}{2}}(1+c_{k})
\end{eqnarray*} 
we thus have 
\begin{eqnarray*} &  \| t^{-\sum_{k=N_{0}+1}^{n}(k - \frac{1}{2})} |\mathcal{M'}_n\rangle - |n-{N_{0}}\rangle \| \le \prod_{k=N_{0}+1}^{n}(c_k+1) - 1 < e^{\sum_{k=N_{0}+1}^\infty c_k} -1 \\ & = e^{t^{-2N_{0}} \frac{st + t^{N_{0}+1/2} + t^{N_{0}+3/2}}{t^2 - 1}}-1 < e^{t^{-N_{0}} \frac{3st^{3/2}}{t^2-1}}-1 \equiv f(s,t)^{t^{-N_{0}}} - 1,
\end{eqnarray*} 
and 
\begin{align*} &\| t^{-\sum_{k=1}^{n}(k - \frac{1}{2})} |\mathcal{M}_n\rangle - \vec{\mathcal{K}}\prod_{k=1}^{N_{0}}(st^{-2(k-\frac{1}{2})}\mathcal{S} + t^{-(k-\frac{1}{2})} + \mathcal{S}^\dagger)|n-N_{0}\rangle \| \\ \le &\| \vec{\mathcal{K}}\prod_{k=1}^{N_{0}}(st^{-2(k-\frac{1}{2})}\mathcal{S} + t^{-(k-\frac{1}{2})} + \mathcal{S}^\dagger) \| \| t^{-\sum_{k=N_{0}+1}^{n}(k - \frac{1}{2})} |\mathcal{M'}_n\rangle - |n-N_{0}\rangle \| 
%\\ < &e^{t^{-N}\frac{3st^{\frac{3}{2}}}{t^2-1}} \prod_{k=1}^N(1+c_k) t^{-N}\\ < &(e^{-t^{-N}\frac{3st^{\frac{3}{2}}}{t^2-1}} - 1)e^{\sum_{k=1}^N}c_k\\  < &(e^{t^{-N}\frac{3st^{\frac{3}{2}}}{t^2-1}} - 1)e^{\frac{3st^{\frac{3}{2}}}{t^2-1}}
\\ < &(f(s,t)^{t^{-N_{0}}} - 1) \prod_{k=1}^{N_{0}}(1+c_k) < (f(s,t)^{t^{-N_{0}}} - 1)e^{\sum_{k=1}^{N_{0}}}c_k  < (f(s,t)^{t^{-N_{0}}} - 1)f(s,t). \end{align*} Let 
\begin{equation} M'_{n,m} = \langle m|\vec{\mathcal{K}}\prod_{k=1}^{N_{0}}(st^{-2(k-\frac{1}{2})}\mathcal{S} + t^{-(k-\frac{1}{2})} + \mathcal{S}^\dagger)|n-N_{0}\rangle, \end{equation} then clearly $M'_{n,m} = 0$ for $m < n - 2{N_{0}}$. 
If we choose 
\begin{equation} N_{0}= \begin{cases} 0 &  f(s,t) < \frac{1+ \sqrt{5}}{2}, \\ -\frac{\log \frac{\log (f^{-1}(s,t) + 1)}{\log f(s,t)}}{\log t}, &\text{otherwise},\end{cases} \end{equation} then 
\begin{equation} \| t^{-\frac{n^2}{2}} |\mathcal{M}_n\rangle - \sum_{m = n - 2 N_{0}}^n M'_{n,m} |m\rangle \| < 1 = M'_{n,n} \le \sup_{m} M'_{n,m}. \end{equation} Therefore $\exists m^* \in [n-2N_{0},n]$, such that $M_{n,m^*} \geq M_{n,m}$ for all $m$. \end{proof}
Lemma \ref{Lemma:mstar} shows that the peak of the $M_{n,m}$ distribution is always within a finite distance from $n$. Essentially,  the bulk of the distribution travels with velocity $1$.

We are now in position to complete the proof of our theorem \ref{Thm:linearEntropy}:
\begin{proof} ({\it Theorem \ref{Thm:linearEntropy}}) 

We separate a linear term from $S_{n}$ as follows (below we supress the $n$ index in $M_{n,m}$):  \begin{eqnarray} & \nonumber S_n = -\sum_{m=0}^{n}s^m p_m\log\frac{M_m^2}{\sum_{m'=0}^ns^{m'}M_{m'}^2}> -\sum_{m=0}^n s^m p_m\log\frac{M_m^2}{s^mM_m^2}\\ \nonumber &= \sum_{m=0}^n s^m p_m m\log s= \sum_{l=0}^n s^{n-l}p_{n-l}(n - l)\log s = \\ &= n\log s - \log s \sum_{l=0}^n\frac{s^{n-l}M_{n-l}^2}{{\sum_{m'=0}^ns^{m'}M_{m'}^2}}l \label{lower entropy bound}
 \end{eqnarray} 
 Taking $m^{*}$ such that $\sup_{m} M_{n,m} = M_{n,m^*}$ and using lemma \ref{Lemma:mstar}, we see that 
 \begin{align*} 
   &   \sum_{l=0}^n\frac{s^{n-l}M_{m^*}^2}{{\sum_{m'=0}^ns^{m'}M_{m'}^2}}l <\sum_{l=0}^n\frac{s^{n-l}M_{m^*}^2}{{ s^{m^{*}}M_{m^{*}}^2}}l= s^{n - m^*} \sum_{l=0}^ns^{-l}l \\ &< s^{2N_{0}} \sum_{l=0}^ns^{-l}l < s^{2N_{0}} \sum_{l=0}^{\infty}s^{-l}l =    \frac{s^{2N_{0}+1}}{(s-1)^2} . \end{align*} Therefore, the remainder term on the right hand side of \eqref{lower entropy bound} is bounded. \end{proof}

\subsection{Entanglement entropy is bounded when $t<1$}

When $t<1$ we expect the Motzkin paths with the lowest area to be exponentially preferred.  In particular, the flat Motzkin path that has zero area has a vanishing contribution to entropy, and thus we expect the entanglement entropy to be substantially reduced. In fact, it turns out that the for any value $t<1$ the entropy is bounded, independently of the size of the system $2n$:

\begin{Thm} When $0<t<1, s\geq 1$, there exists a constant $C(s,t)$ independent of the system size $n$, such that for any $n$, $S_n < C(s,t)$. \label{tless1}\end{Thm}
\noindent {\it Remark:} Note that theorem holds both for the colored and uncolored case $s=1$.

For the theorem to hold,  the exponential growth in contribution to entropy from the possible colorings of higher paths should be overwhelmed by the exponential price in area. Technically, we need the quantities $M_m^{2}$ to decrease faster than the rate $s^m$ grows in order to make $p_m$ decrease exponentially.

%Actually, for $t>1$, the exponentially increasing factor $s^m$ in $p_m$ is already enough to give linear scaling of entropy even if $M_m$ is uniform, let alone it enhances $p_m$ with larger $m$ by reaching its maximum finite steps away from $n$. However, when $t<1$, we need $M_m$ to decrease not only exponentially, but faster than $s^m$ increase as well so as to make $p_m$ decrease exponentially. So 

To highlight this feature we first define \begin{equation} \tilde{M}_{n,m}=s^{\frac{m}{2}}M_{n,m}~~,~~\tilde{p}_{n,m}=\frac{\tilde{M}_{n+1,m}^2}{\sum_{m=0}^n \tilde{M}_{n+1,m}^2}.\end{equation} Substitution into \eqref{recurrel} gives the relation
\begin{align} %s^{-\frac{m}{2}}\tilde{M}_{n+1,m} &= s^{-\frac{m+1}{2}}s t^{m+\frac{1}{2}}\tilde{M}_{n,m+1} + s^{-\frac{m}{2}}t^m\tilde{M}_{n,m} + s^{-\frac{m-1}{2}} t^{m-\frac{1}{2}}\tilde{M}_{n,m-1}\\ 
\tilde{M}_{n+1,m} &= \sqrt{s} t^{m+\frac{1}{2}}\tilde{M}_{n,m+1} + t^m\tilde{M}_{n,m} + \sqrt{s} t^{m-\frac{1}{2}}\tilde{M}_{n,m-1}\label{tildeRec}, \end{align} for $m\in [1,n-1]$.\\

To prove the entropy is bounded, we need the following lemmas.

\begin{Lem} \begin{equation} \sum_m \tilde{M}_{n+1,m}^2 > \sum_m \tilde{M}_{n,m}^2. \end{equation}\label{Lem:MonotonicDenom} \end{Lem}

\begin{proof} From \eqref{korder}, we have \begin{align*} |\mathcal{M}_{n+1}\rangle &= \vec{\mathcal{K}}\prod_{k=1}^{n+1}  (st^{-(k - \frac{1}{2})} \mathcal{S} + 1 + t^{k - \frac{1}{2}} \mathcal{S}^\dagger) |0\rangle \\ &= \vec{\mathcal{K}}\prod_{k=1}^n  (st^{-(k - \frac{1}{2})} \mathcal{S} + 1 + t^{k - \frac{1}{2}} \mathcal{S}^\dagger) (st^{-(n + \frac{1}{2})} \mathcal{S} + 1 + t^{n + \frac{1}{2}} \mathcal{S}^\dagger)|0\rangle, \\ &= |\mathcal{M}_n\rangle + \vec{\mathcal{K}}\prod_{k=1}^n  (st^{-(k - \frac{1}{2})} \mathcal{S} + 1 + t^{k - \frac{1}{2}} \mathcal{S}^\dagger) (st^{-(n + \frac{1}{2})} \mathcal{S} + t^{n + \frac{1}{2}} \mathcal{S}^\dagger)|0\rangle.\end{align*} The last term on the RHS of the equation contains non-zero contributions for all states $|m\rangle$, with $m=0,..n+1$, and we have: \begin{align*} M_{n+1,m} &> M_{n,m},\\ \tilde{M}_{n+1,m} &> \tilde{M}_{n,m} \qquad \forall m\geq 0,n\geq 1. \end{align*} And the Lemma follows. \end{proof}
Next we establish the following bound on $\tilde{p}_{n,m}$:
\begin{Lem} \begin{equation} \tilde{p}_{n,m} < 9 s t^{2m-1}.\end{equation}  \label{Lemma:pnm bound} \end{Lem}

\begin{proof} By definition of $\tilde{p}_{n,m}$, and the recursion relation \eqref{tildeRec}, 
\begin{align*} \tilde{p}_{n,m} &= \frac{\tilde{M}_{n+1,m}^2}{\sum_{m=0}^n \tilde{M}_{n+1,m}^2} = \frac{(\sqrt{s} t^{m+\frac{1}{2}}\tilde{M}_{n,m+1} + t^m\tilde{M}_{n,m} + \sqrt{s} t^{m-\frac{1}{2}}\tilde{M}_{n,m-1})^2}{\sum_{m=0}^n \tilde{M}_{n+1,m}^2} \\ &= \frac{t^{2m}(\sqrt{s} t^{\frac{1}{2}}\tilde{M}_{n,m+1} + \tilde{M}_{n,m} + \sqrt{s} t^{-\frac{1}{2}}\tilde{M}_{n,m-1})^2}{\sum_{m=0}^n \tilde{M}_{n+1,m}^2} \\ &\le t^{2m} \frac{(3 \sqrt{s\over t}max\{t \tilde{M}_{n,m+1}, \sqrt{\frac{t}{s}} \tilde{M}_{n,m}, \tilde{M}_{n,m-1}\})^2}{\sum_{m=0}^n \tilde{M}_{n+1,m}^2} \\ &\le 9 t^{2m} \frac{s}{t} \frac{max\{ \tilde{M}_{n,m+1}^2, \tilde{M}_{n,m}^2, \tilde{M}_{n,m-1}^2\}}{\sum_{m=0}^n \tilde{M}_{n+1,m}^2} < 9 t^{2m} {\frac{s}{t}}. 
\end{align*} Lemma \ref{Lem:MonotonicDenom} was used in the last line.\end{proof}

We now have the ingredients to prove theorem \ref{tless1}:
\begin{proof} ({\it Theorem  \ref{tless1}})  Using Lemma \ref{Lemma:pnm bound} we see that when \begin{equation} m > m_0 \equiv \Big[\frac{\log(\frac{1}{9e} {\frac{t}{s}})}{2 \log{t}}\Big] + 1, \end{equation} we have \begin{equation} \tilde{p}_{n,m} < 9 s t^{2m-1}<{1\over e}. \end{equation} 
It is easy to check that the function $-x\log(x)$ is monotonically increasing when $x\in (0, \frac{1}{e})$, in other words, for $m>m_0$, 
\begin{eqnarray} \tilde{p}_{n,m}< 9 s t^{2m-1}<{1\over e}~\Longrightarrow~ -\tilde{p}_{n,m}\log \tilde{p}_{n,m} < -9 s t^{2m-1} \big(\log( {\frac{9s}{t}}) + 2m\log t\big). \end{eqnarray} 
Therefore \begin{align*} S_n &= -\sum_{m=0}^n \tilde{p}_{n,m}\log \tilde{p}_{n,m} + \log s \sum_{m=0}^n \tilde{p}_{n,m} m\\ &< -\sum_{m=0}^{m_0} \tilde{p}_{n,m}\log \tilde{p}_{n,m} -\sum_{m=m_0+1}^\infty 9 s t^{2m-1}   \big(\log( {\frac{9s}{t}}) + 2m\log t\big) + \log s \sum_{m=0}^\infty  9 s t^{2m-1} m\\ &< \frac{m_0+1}{e} - \frac{9s t^{2m_0+1}}{1-t^2}\log( {\frac{9s}{t}}) - \frac{18 s t^{2 m_0+2} (m_0(1- t^2) +1)}{(t^2-1)^2} \log t\\ &\quad+ \frac{9s t}{(t^2-1)^2} \log s \equiv C(s,t),
 \end{align*}
where we used $\sup_{x\in(0,1)}-x log(x)=e^{-1}$ for entropy terms with $m\leq m_{0}$ in the last inequality.
\end{proof}

\section{Summary and Open Questions}
%Bravyi and Gosset in [Gosset and Bravyi] provided a complete classification of gapped and gapless phases of frustration-free translation-invariant spin-1/2 chains with nearest-neighbor interactions and stated that it was a challenging open problem to do the same thing for qudits, or spin chains with a Hilbert space of dim $ d\geq3 $. Our model may potentially be the very first step toward generalizing the results in [Gosset and Bravyi] to qudits and constructing a phase diagram for $d$-dimensional spin chains with $d\ge 3$.

%The scaling behavior of the variance and that of the entanglement entropy,  was shown to be the same in many different systems []. The scaling behavior of the fluctuations of a random variable are often represented through the Full Counting Statistics (FCS) method. In ~\cite{RS2012UnColored} as well as in ~\cite{movassagh2015power}, the random variable is the height of the Motzkin path at the middle of the spin chain and the probability density is that which they derived for the height of the path. 
In this paper we have presented a continuous family of Hamiltonians with an exactly solvable, frustration free and non-degenerate ground state. In the colored model, the family features an exotic phase transition between volume and boundary entropies with transitions through a $\sqrt{n}$ area law violation, that has an intuitive interpretation in terms of Motzkin paths. 
One can apply the ideas presented herein, i.e. searching for quantum phase transitions associated with frustration free deformation, into other interesting frustration-free models. In particular the model of   \cite{salberger2016fredkin} may be directly amenable to an analogous continuous deformation, where Dyck paths are weighted rather than Motzkin paths to obtain volume scaling of entanglement entropy. Other important questions include the search for translationally invariant Hamiltonians with similar properties: How to get rid of the boundary terms while leaving the Hamiltonian frustration free and non-degenerate? It is also of interest to explore whether there is a way to make the entanglement entropy scale with a larger linear coefficient. More precisely, for $t>1$, we find that $S_n\propto log({d-1\over 2})n$ where $ d$ is the dimension of the Hilbert space of each individual spin. Is it possible to construct a model where $S_n\propto c n$ with $c>log({d-1\over 2}) $ without increasing the interaction range in the Hamiltonian?

While we have been mostly interested in entanglement scaling, it is worthwhile to study more aspects of the model. In particular, the scaling of the spectral gap with the length of the chain is an important quantity that indicates how fast the system becomes gapless in the thermodynamic limit. It has been shown in ~\cite{movassagh2015power} that the spectral gap for the colored Motzkin path model has an upper bound that scales as $ O(n^{-2}) $ and a lower bound that scales as $ O(n^{-c}) $ where $ c>1 $. Since in our model, the entanglement entropy scales linearly, we expect the upper bound of the spectral gap to scale even faster. In addition, we expect that our system opens a gap  in the region where the entropy is $O(1)$, however, this conjecture requires a separate treatment (the entanglement entropy being bounded does not imply that the Hamiltonian is gapped).

\section{Acknowledgements}
It is a pleasure to thank R. Movassagh, P. Fendley and P. Arnold for valuable discussions.
The work of IK was supported by NSF grant DMR-1508245.
\bibliographystyle{unsrt}
\bibliography{Motzkin}

\begin{thebibliography}{10}

\bibitem{movassagh2015power}
Ramis Movassagh and Peter~W Shor.
\newblock Power law violation of the area law in quantum spin chains.
\newblock {\em arXiv preprint arXiv:1408.1657}, 2015.

\bibitem{page1993average}
Don~N Page.
\newblock Average entropy of a subsystem.
\newblock {\em Physical review letters}, 71(9):1291, 1993.

\bibitem{foong1994proof}
SK~Foong and S~Kanno.
\newblock Proof of page's conjecture on the average entropy of a subsystem.
\newblock {\em Physical review letters}, 72(8):1148, 1994.

\bibitem{sen1996average}
Siddhartha Sen.
\newblock Average entropy of a quantum subsystem.
\newblock {\em Physical review letters}, 77(1):1, 1996.

\bibitem{laflorencie2015quantum}
Nicolas Laflorencie.
\newblock Quantum entanglement in condensed matter systems.
\newblock {\em arXiv preprint arXiv:1512.03388}, 2015.

\bibitem{hastings2007area}
Matthew~B Hastings.
\newblock An area law for one-dimensional quantum systems.
\newblock {\em Journal of Statistical Mechanics: Theory and Experiment},
  2007(08):P08024, 2007.

\bibitem{arad2013area}
Itai Arad, Alexei Kitaev, Zeph Landau, and Umesh Vazirani.
\newblock An area law and sub-exponential algorithm for 1d systems.
\newblock {\em arXiv preprint arXiv:1301.1162}, 2013.

\bibitem{eisert2010colloquium}
Jens Eisert, Marcus Cramer, and Martin~B Plenio.
\newblock Colloquium: Area laws for the entanglement entropy.
\newblock {\em Reviews of Modern Physics}, 82(1):277, 2010.

\bibitem{callan1994high}
C.~Callan and F.~Wilczek.
\newblock On geometric entropy.
\newblock {\em Phys. Lett. B 333, 55}, 55:61, 1994.

\bibitem{holzhey1994high}
C.~Holzhey, F.~Larsen, and F.~Wilczek.
\newblock {High Energy Physics-Theory Title: Geometric and Renormalized Entropy
  in Conformal Field Theory}.
\newblock {\em Journal reference: Nucl. Phys. B424}, 443:467, 1994.

\bibitem{calabrese2009entanglement}
Pasquale Calabrese and John Cardy.
\newblock Entanglement entropy and conformal field theory.
\newblock {\em Journal of Physics A: Mathematical and Theoretical},
  42(50):504005, 2009.

\bibitem{wolf2006violation}
M.M. Wolf.
\newblock {Violation of the entropic area law for fermions}.
\newblock {\em Physical Review Letters}, 96(1):10404, 2006.

\bibitem{gioev2006entanglement}
D.~Gioev and I.~Klich.
\newblock {Entanglement entropy of fermions in any dimension and the Widom
  conjecture}.
\newblock {\em Phys. Rev. Lett.}, 96(10):100503, 2006.

\bibitem{gottesman2010entanglement}
Daniel Gottesman and MB~Hastings.
\newblock Entanglement versus gap for one-dimensional spin systems.
\newblock {\em New journal of physics}, 12(2):025002, 2010.

\bibitem{irani2010ground}
Sandy Irani.
\newblock Ground state entanglement in one-dimensional translationally
  invariant quantum systems.
\newblock {\em Journal of Mathematical Physics}, 51(2):022101, 2010.

\bibitem{vitagliano2010volume}
G~Vitagliano, A~Riera, and JI~Latorre.
\newblock Volume-law scaling for the entanglement entropy in spin-1/2 chains.
\newblock {\em New Journal of Physics}, 12(11):113049, 2010.

\bibitem{ramirez2014conformal}
Giovanni Ram{\'\i}rez, Javier Rodr{\'\i}guez-Laguna, and Germ{\'a}n Sierra.
\newblock From conformal to volume law for the entanglement entropy in
  exponentially deformed critical spin 1/2 chains.
\newblock {\em Journal of Statistical Mechanics: Theory and Experiment},
  2014(10):P10004, 2014.

\bibitem{salberger2016fredkin}
Olof Salberger and Vladimir Korepin.
\newblock Fredkin spin chain.
\newblock {\em arXiv preprint arXiv:1605.03842}, 2016.

\bibitem{bravyi2012criticality}
Sergey Bravyi, Libor Caha, Ramis Movassagh, Daniel Nagaj, and Peter~W Shor.
\newblock Criticality without frustration for quantum spin-1 chains.
\newblock {\em Physical review letters}, 109(20):207202, 2012.

\bibitem{movassagh2010unfrustrated}
Ramis Movassagh, Edward Farhi, Jeffrey Goldstone, Daniel Nagaj, Tobias~J
  Osborne, and Peter~W Shor.
\newblock Unfrustrated qudit chains and their ground states.
\newblock {\em Physical Review A}, 82(1):012318, 2010.

\end{thebibliography}

\end{document}